\documentclass{amsart}

\usepackage{amsmath,amsthm,amssymb}
\usepackage{graphicx}
\usepackage{tikz-cd}
\usepackage{paralist}
\usepackage{environ}
\usepackage{xcolor}
\usepackage{hyperref}
\usepackage{centernot}
\usepackage{mathtools}
\usepackage{marvosym}
\usepackage{xcolor,cancel}

\usepackage{verbatim}

\usepackage{pstricks}

\usepackage{mdframed}

\usepackage{subfigure}
\renewcommand{\thesubfigure}{\thefigure.\arabic{subfigure}}
\makeatletter
\renewcommand{\p@subfigure}{}
\renewcommand{\@thesubfigure}{\thesubfigure:\hskip\subfiglabelskip}
\makeatother

\makeatletter
\def\square{\pst@object{square}}
\def\square@i(#1,#2)#3{{\use@par\solid@star\psframe[origin={#1,#2}](#3,#3)}}
\makeatother

\makeatletter
\DeclareFontFamily{U}{tipa}{}
\DeclareFontShape{U}{tipa}{bx}{n}{<->tipabx10}{}
\newcommand{\arc@char}{{\usefont{U}{tipa}{bx}{n}\symbol{62}}}%

\newcommand{\arc}[1]{\mathpalette\arc@arc{#1}}

\newcommand{\arc@arc}[2]{%
  \sbox0{$\m@th#1#2$}%
  \vbox{
    \hbox{\resizebox{\wd0}{\height}{\arc@char}}
    \nointerlineskip
    \box0
  }%
}
\makeatother

\makeatletter
\newcommand{\doublewedge}{\big@doubleop{\wedge}}
\newcommand{\big@doubleop}[1]{%
  \DOTSB\mathop{\mathpalette\big@doubleop@aux{#1}}\slimits@
}

\newcommand\big@doubleop@aux[2]{%
  \sbox\z@{$\m@th#1#2$}%
  \makebox[1.35\wd\z@][s]{$\m@th#1#2\hss#2$}%
}

\newcommand{\abs}[1]{\left|#1\right|}     

\newcommand{\rnear}{\widetilde{\delta}_{\Phi}} 

\theoremstyle{plain}
\newtheorem{theorem}{Theorem}

\newtheorem{remark}{Remark}
\newtheorem{definition}{Definition}

\newtheorem{example}{Example}

\makeatletter
\@namedef{subjclassname@2020}{\textup{2020} Mathematics Subject Classification}
\makeatother

\begin{document}

\title{Characteristically Near Stable Vector Fields\\ in the Polar Complex Plane}

\author[J.F. Peters]{J.F. Peters}
\address{
Department of Electrical and Computer Engineering,
University of Manitoba, WPG, Manitoba, R3T 5V6, Canada and
Department of Mathematics, Faculty of Arts and Sciences, Ad\.{i}yaman University, 02040 Ad\.{i}yaman, Turkey
}
\email{james.peters3@umanitoba.ca}

\author[E. Cui]{E. Cui}
\address{
Department of Electrical and Computer Engineering,
University of Manitoba, WPG, Manitoba, R3T 5V6, Canada
}
\email{cuie@myumanitoba.ca}

\thanks{The research has been supported by the Natural Sciences \&
Engineering Research Council of Canada (NSERC) discovery grant 185986 
and Instituto Nazionale di Alta Matematica (INdAM) Francesco Severi, Gruppo Nazionale per le Strutture Algebriche, Geometriche e Loro Applicazioni grant 9 920160 000362, n.prot U 2016/000036 and Scientific and Technological Research Council of Turkey (T\"{U}B\.{I}TAK) Scientific Human
Resources Development (BIDEB) under grant no: 2221-1059B211301223.}



%

%


\date{}

\begin{abstract}
This paper introduces results for characteristically near vector fields that are stable or non-stable in the polar complex plane $\mathbb{C}$.  All characteristic vectors (aka eigenvectors) emanate from the same fixed point in $\mathbb{C}$, namely, 0. Stable characteristic vector fields satisfy an extension of the Krantz stability condition, namely, the maximal eigenvalue of a stable system lies within or on the boundary of the unit circle in $\mathbb{C}$. 
\end{abstract}
%

\maketitle
\tableofcontents



%




\section{Introduction}
This paper introduces proximities of characteristic vector fields that are stable in the polar complex plane.
A dynamical system is a 1-1 mapping from a set of points $M$ to itself~\cite[\S 9.1.1]{Krantz2010}, which describes the time-dependence of a point in a complex ambient system.  In its earliest incarnation by Poincar\'{e}, the focus was on the stability of the solar system~\cite{Poicare1890}. More recently, dynamical system behaviour is in the form of varying oscillations in motion waveforms~\cite{DeLeoYork2024,Feldman2011vibration}. Typically, vector fields are used to construct dynamical systems (see, e.g.,~\cite[\S 4]{Pokorny2008},~\cite{Forman1998}).  

The focus here is on dynamical systems generated by stable characteristic vector fields (cVfs) in $\mathbb{C}$ and their corresponding semigroups. Comparison of cVf characteristics leads to the detection of proximal cVf semigroups.  In general, a characteristic of an object $X$ is a mapping $\varphi:X\to \mathbb{C}$ with values $\varphi(x\in X)$ that provide an object profile.  Proximal objects $X,Y$ require \boxed{\abs{\varphi(x\in X)-\varphi(y\in Y)}\in [0,1]}.  All characteristic vectors (aka eigenvectors) emanate from the same fixed point in $\mathbb{C}$, namely, 0. Stable characteristic vector fields satisfy the Krantz stability condition, namely, all eigenvalues lie inside the unit circle in $\mathbb{C}$.
 
An application of the proposed approach is given in measuring system stability in terms of vector fields emanating from oscillatory waveforms derived from the up-and-down movements of a walker, runner or biker recorded in a sequence of infrared video frames.   We prove that system stability occurs when its maximum eigenvalue occurs within or on the boundary of the unit circle in complex plane (See Theorem~\ref{theorem:stability condition}). This result extends results in~\cite{Tiwari2024},~\cite{HaiderPeters2021self-similar}) as well as in~\cite{PetersVergili2023,OzkanKulogluPeters2021self-similarity,ErdagPetersDeveci2024Dnear0}.

\begin{figure}[!ht]
	\centering
\includegraphics[width=125mm]{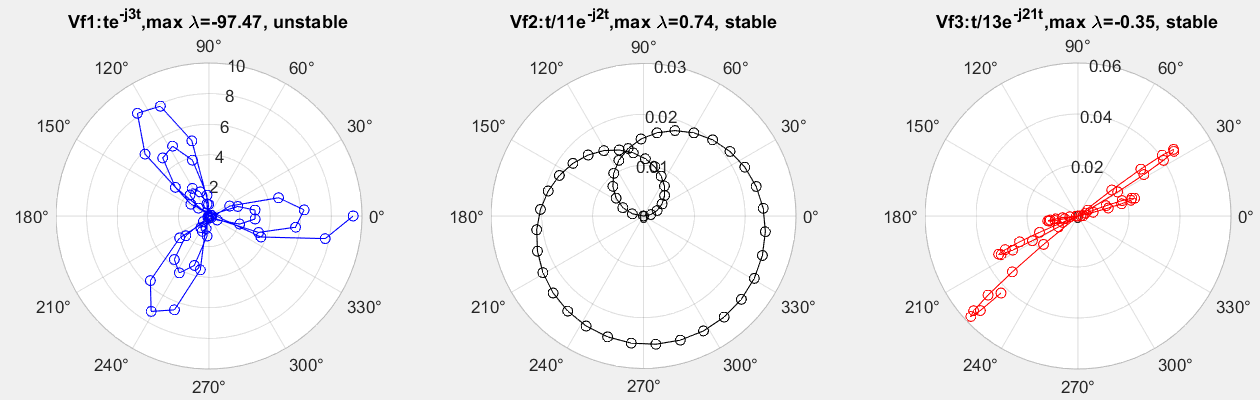}
  \caption{Three Vector Fields in Polar Complex Plane:
	\boxed{\mbox{(leftmost,unstable)}\ \vec{V}f_1},\boxed{\mbox{(middle,stable)}\ \vec{V}f_2},\boxed{\mbox{(rightmost,stable)}\ \vec{V}f_3}}
	\label{fig:polarZ}
\end{figure}

\begin{table}[h!]\label{table:symbols}
\begin{center}
\caption{Principal Symbols Used in this Paper}
\begin{tabular}{|c|c|}
\hline
Symbol & Meaning\\ 
\hline\hline
$\mathbb{C}$ & Complex plane.\\
\hline
$j$ & $j^2 = -1\ \mbox{(imaginary number)}$.\\
\hline
$\vec{0}$ & center of unit circle in polar $\mathbb{C}$.\\
\hline 
$z$ & = $a + jb=e^{j\theta},a,jb\in\mathbb{C}$ (see Fig.~\ref{fig:C}).\\
\hline
$2^X$ & $\mbox{collection of subsets in $X$}$.\\
\hline
$A\ \rnear\ B$ & $A$ characteristically near $B$.\\
\hline
$\varphi(a\in A)\in \mathbb{C}$ & Characteristic of $a\in A$.\\
\hline
$\Phi(A)=$ & $\left\{\varphi(a_1),\dots,\varphi(a_n):a_1,\dots,a_n\in A \right\}\in2^{\mathbb{C}}$.\\
\hline
$d^{\boldsymbol{\widetilde{\Phi}}}(A,B)$ &  Characteristic Distance.\\
\hline
\end{tabular}
\end{center}
\end{table}

\begin{figure}[!ht]
	\centering
\includegraphics[width=74mm]{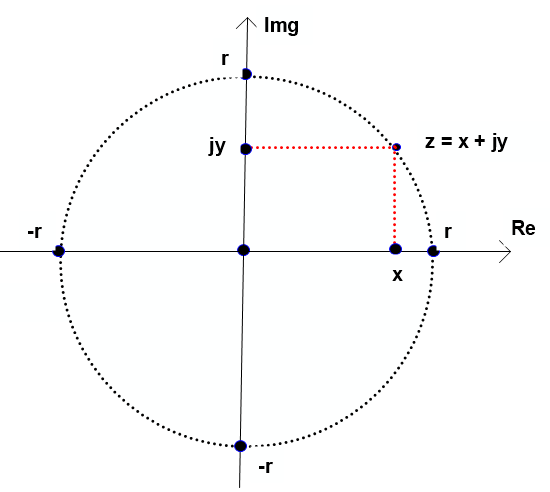}
\caption{Complex Plane.}
	\label{fig:C}
\end{figure}

\section{Preliminaries}
\noindent Detected affinities between vector fields for stable systems results from determining the infimum of the distances between pairs of system characteristics.  

\begin{definition}\label{def:vector}{\bf (Vector).}\\
A {\bf vector} $v$ (denoted by $\vec{v}$) is a quantity that has magnitude and direction in the complex plane $\mathbb{C}$.
\qquad \textcolor{blue}{$\blacksquare$}
\end{definition}

\begin{definition}\label{def:Vf}{\bf (Vector Field in the Complex Plane).}\\
Let $U = \left\{p\in \mathbb{C}\right\}$ be a bounded region in the complex plane containing points $p(x,jy)\in U$. A {\bf vector field} is a mapping $F:U\to 2^{\mathbb{C}}$ defined by 
\begin{center}
\boxed{\boldsymbol{
F(p(x,jy)) = \left\{\vec{v}\right\}\in 2^{\mathbb{C}}\ \mbox{denoted by}\ \vec{V}f.\ 
\mbox{\qquad \textcolor{blue}{$\blacksquare$}}
}}
\end{center}
\end{definition}

\vspace*{0.2cm}

\begin{remark}
A complex number $z$ in polar form (discovered by Euler~\cite{Euler1748}) is written \boxed{z=re^{j\theta}}.
\qquad \textcolor{blue}{$\blacksquare$}
\end{remark}

\vspace*{0.2cm}
\begin{example}
Three examples of vector fields in polar form are given in Fig.~\ref{fig:polarZ}. 
\qquad \textcolor{blue}{$\blacksquare$}
\end{example}

\vspace*{0.2cm}

\begin{remark}
A complex number $z$ in polar form (discovered by Euler~\cite{Euler1748}) is written \boxed{z=re^{j\theta}}.
\qquad \textcolor{blue}{$\blacksquare$}
\end{remark}

\vspace*{0.2cm}

\begin{definition}\label{def:Vf}{\bf (Vector Field in the Complex Plane).}\\
Let $U = \left\{z\in \mathbb{C}\right\}$ be a bounded region in the complex plane containing points $z(x,jy)\in U\subset\mathbb{C}$. A {\bf vector field} is a mapping $F:U\to 2^{\mathbb{C}}$ defined by 
\begin{center}
\boxed{\boldsymbol{
F(z(x,yj)) = \left\{\vec{v}\in 2^{\mathbb{C}}\right\}\ \mbox{denoted by}\ \vec{V}f.\ 
\mbox{\qquad \textcolor{blue}{$\blacksquare$}}
}}
\end{center}
\end{definition}

\begin{definition}\label{def:eig}{\bf (Eigenvalue $\boldsymbol{\lambda}$(aka Characteristic value)).}\\
The eigenvalues (characteristic values) of a matrix $A$ are solutions to the determinant \boxed{\bf det(A-\lambda I), I = \begin{bmatrix}1 & 0\\0 & 1 \end{bmatrix}} identity matrix.  \qquad \textcolor{blue}{$\blacksquare$}
\end{definition}

\begin{example}{\bf (Sample Eigenvalues).}\\
$A =\begin{bmatrix}4 & 8\\6 & 26\end{bmatrix}$,
$I =\begin{bmatrix}1 & 0\\0 & 1 \end{bmatrix}$:
$det(A - \lambda I) = 
\left|
\begin{array}{cc}
4-\lambda & 8\\6 & 26-\lambda
\end{array}
\right|
= (4-\lambda)(26-\lambda)-(8)(6)= 0$\\
$104-30\lambda+\lambda^2-48 =
\lambda^2-30\lambda+56 = 
(\lambda-28)(\lambda-2) = 0$\\
\boxed{\boldsymbol{
\lambda_1 = 28, \lambda_2 = 2\ \mbox{\bf(eigen values of)}\ A
}\mbox{\qquad \textcolor{blue}{$\blacksquare$}}}
\end{example}

\vspace*{0.2cm}

\begin{definition}\label{def:eigenv}{\bf (Eigenvector).}\\
Given a matrix $A$, then $\vec{v}$ is a eigenvector, provided
\begin{center}
\boxed{\boldsymbol{
A\vec{v}-\lambda\vec{v}=0\in \mathbb{C}.
\ \mbox{\qquad \textcolor{blue}{$\blacksquare$}}
}}
\end{center}
\end{definition}

\vspace*{0.2cm}

\begin{table}[htbp]\label{table:eigv}
\centering
\caption{Eigenvectors derived from \boxed{\frac{t}{11}e^{j2t}} Vf}
\begin{tabular}{|c|c|c|}
\hline
$1^{st}\ \mathbb{C}$ \textbf{quadrant}  &  $1^{st}\ \mathbb{C}$ \textbf{quadrant} & $3^{rd}\ \mathbb{C}$ \textbf{quadrant}\\
\hline
\hline
$\vec{z_{11}}$=0.1500+0.0498j  & $\vec{z_{12}}$=0.0106+0.0035j & $\vec{z_{13}}$=-0.0754-0.0250j \\
\hline
$\vec{z_{21}}$=0.1586+0.0471j  & $\vec{z_{22}}$=0.0333+0.0091j & $\vec{z_{23}}$=-0.0418-0.0124j \\
\hline
\end{tabular}
\label{table:eigv2}
\end{table}

\vspace*{0.2cm}

\begin{example}{\bf (Sample eigenvectors in center \boxed{\frac{t}{11}e^{j2t}} $\vec{V}$f in Fig.~\ref{fig:polarZ}).}\\
A selection of eigenvectors from the first and third quadrants in the polar complex plane in the center vector field in Fig.~\ref{fig:polarZ} are given in Table~\ref{table:eigv2}.
\qquad \textcolor{blue}{$\blacksquare$}
\end{example}

\vspace*{0.2cm}

\begin{definition}\label{def:Krantz}{\bf (Krantz Vector Field Stability Condition~\cite{Krantz2010}).}\\
A vector field $\vec{V}f$ in the complex plane is stable, provided all of the eigenvalues of $\vec{V}f$ are either within or on the boundary of the unit circle centered $\boldsymbol{0}$ in $\mathbb{C}$. \qquad \textcolor{blue}{$\blacksquare$}
\end{definition}

\begin{mdframed}[backgroundcolor=green!15]
\begin{theorem}\label{theorem:stability condition}{\bf (Vector Field Stability Condition).}\\
A vector field $\vec{V}f$ in the complex plane is stable, provided the maximal eigenvalue of $\vec{V}f$ lies within or on the boundary of the unit circle in $\mathbb{C}$.
\end{theorem}
\begin{proof}
From Def.~\ref{def:Krantz}, all eigenvalues $D=\left\{\lambda\right\}$ for a stable vector field lie either within or on the boundary of the unit circle in $\mathbb{C}$. Hence, \boxed{max(\lambda)\in D} lies either within or on the boundary of the unit circle in $\mathbb{C}$.
\end{proof}
\end{mdframed}

\vspace*{0.2cm}

\begin{table}[htbp]\label{table:eigval}
\centering
\caption{Eigenvalues derived from \boxed{\frac{t}{11}e^{j2t}} Vf}
\begin{tabular}{|c|c|c|c|c|}
\hline
\hline
$\lambda_{max}$=-0.7384  &
$\lambda_{max-1}$=-0.2328 & 
$\lambda_{max-2}$=-0.0823 &
$\lambda_{max-3}$=-0.0488 &
$\lambda_{max-4}$=-0.0298\\
\hline
\end{tabular}
\label{table:eigval2}
\end{table}

\vspace*{0.2cm}

\begin{figure}[!ht]
	\centering
\includegraphics[width=95mm]{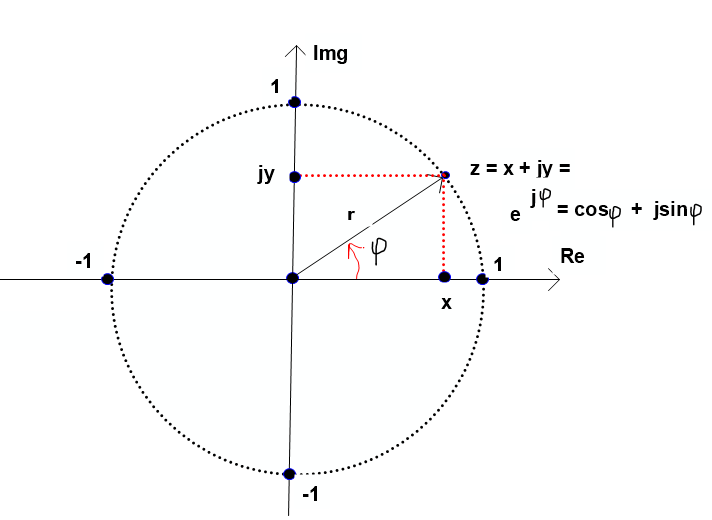}
  \caption{Polar Complex Plane}
	\label{fig:polarZ}
\end{figure}

\begin{example}{\bf (Largest \boxed{\lambda} values for the center \boxed{\frac{t}{11}e^{j2t}} vector field in Fig.~\ref{fig:polarZ}).}\\
The 5 bigest eigenvalues derived from the center vector field \boxed{Vf} in Fig.~\ref{fig:polarZ} in the polar complex plane are given in Table~\ref{table:eigval2}.  From Theorem~\ref{theorem:stability condition}, \boxed{Vf} is stable, since $\lambda_{max}$=-0.7384 in Table{table:eigval2} lies within the unit circle in the polar complex plane $\mathbb{C}$ (see Fig.~\ref{fig:polarZ}).
\qquad \textcolor{blue}{$\blacksquare$}
\end{example}

\vspace*{0.2cm}

\begin{definition}
A characteristic of an object (aka sets, systems) $X$ is a mapping $\varphi$:
\begin{center}
\boxed{\boldsymbol{
\varphi:X\to \mathbb{C}\ \mbox{defined by}\ \varphi(x\in X) \in \mathbb{C}.
}}
\end{center}
\end{definition}
\vspace*{0.2cm}

\begin{definition}\label{def:dnear0}{\bf (Characteristic Distance).}\\  
Let $X,Y$ be nonempty sets and $a\in A\in 2^X,b\in B\in 2^Y$ and let $\varphi(a),\varphi(b)$ be numerical characteristics inherent in $A$ and $B$.  
The nearness mapping $d^{\Phi}: 2^X\times 2^Y\to\mathbb{R}$ is defined by
	\begin{center}
		\boxed{d^{\Phi}(A,B)=\inf_{\substack{a\in A \\ b\in B}} \left\{\varphi(a) - \varphi(b)\right\} = \varepsilon\in[-1,1]]\in\mathbb{C}.\mbox{	\qquad \textcolor{blue}{$\blacksquare$}}
	}
	\end{center}		
\end{definition}

\begin{remark}{\bf (Characteristically Near Objects).}\\
\noindent Objects $A,B$ are characteristically near,  provided $d^{\Phi}(A,B) = 0$ for some $a\in A, b\in B$.
\qquad \textcolor{blue}{$\blacksquare$}
\end{remark}

\vspace*{0.2cm}

\begin{definition}\label{def:stability}{\bf (Stability View of a Vector Field-Based Dynamical System).}\\
A characteristic of an object (aka sets, systems) $X$ is a mapping $\varphi$:
\begin{center}
\boxed{\boldsymbol{
\varphi:X\to \mathbb{C}\ \mbox{defined by}\ \varphi(x\in X) \in \mathbb{C}.
}}
\end{center}
\end{definition}

\vspace*{0.2cm}

\begin{definition}\label{def:nearness}{\bf (Characteristic Nearness of Systems~\cite{peters2025}).}\\
Let $X,Y$ be a pair of systems.  For nonempty subsets $A\in 2^X, B\in 2^Y$, the characteristic nearness of $A,B$ (denoted by $A\ \rnear\ B$) is defined by
\begin{center}
\boxed{\boldsymbol{
A\ \rnear\ B \Leftrightarrow
d^{\boldsymbol{\large\widetilde{\Phi}}}(A,B)=0.
\mbox{\qquad \textcolor{blue}{$\blacksquare$}}  
}}
\end{center}
\end{definition}

\vspace*{0.2cm}

\begin{mdframed}[backgroundcolor=green!15]
\begin{theorem}\label{theorem:nearness}{\bf (Fundamental Theorem of Near Systems)}.\\
Let $X,Y$ be a pair of systems with $A\in 2^X,B\in 2^Y$.

\begin{center}
\boxed{\boldsymbol{
A\ \rnear\ B\ \Leftrightarrow\ \exists a\in A,b\in B:
\abs{\varphi(a)-\varphi(b)} = 0.
}}
\end{center}
\end{theorem}
\begin{proof}$\mbox{}$\\
\boxed{\Rightarrow}: From Def.~\ref{def:dnear0}, $A\ \rnear\ B$ implies that there is at least one pair $a\in A, b\in B$ such that $d^{\Phi}(A,B)=\abs{\varphi(a) - \varphi(b)}  = 0$, i.e, when $a=b$.\\
\noindent \boxed{\Leftarrow}:
Given $d^{\Phi}(A,B) =0$, we know that $\inf_{\substack{a\in A \\ b\in B}} \abs{\varphi(a) - \varphi(b)} = 0\in\mathbb{C}$. Hence, from Def.~\ref{def:nearness}, $A\ \rnear\ B$, also.  
That is, sufficient nearness of at least one pair characteristics $\varphi(a\in A), \varphi(b\in B)\in [0,1]\in \mathbb{C}$ indicates the characteristic nearness of the sets, i.e., we conclude $A\ \rnear\ B$. 
\end{proof}

\begin{theorem}\label{theorem:charCloseSystems}{\bf (Characteristically Close Systems).}\\
Systems $X,Y$ are characteristically near if and only $X,Y$ contain subsystems that are characteristically near.
\end{theorem}
\begin{proof}
Immediate from Theorem~\ref{theorem:nearness}.
\end{proof}

\begin{theorem}\label{theorem:nearStableSystems}
{\bf (Stable Systems Extreme Closeness Condition).}\\
Let $\vec{V}f1,\vec{V}f2$ be vector fields representing a pair of stable systems and let 
$max\lambda_{vec{V}f1},max\lambda_{vec{V}f2}$ be the maximum $\lambda$ (eigenvalues) for the pair of systems. If
\boxed{\abs{max\lambda_{\vec{V}f1}-max\lambda_{\vec{V}f2}}\in [0,0.5]}, then $\vec{V}f1\ \rnear\ \vec{V}f2$.
\end{theorem}
\begin{proof}
From Theorem~\ref{theorem:stability condition}, for the vector field $\vec{V}f$ for a stable system, \\
\boxed{max\lambda_{vec{V}f}\in [0,\pm 1]}. 

For a pair of system vector fields $\vec{V}f1,\vec{V}f2$, assume that
\begin{center}   
\boxed{\boldsymbol{
\abs{max\lambda_{\vec{V}f1}-max\lambda_{\vec{V}f2}}\in [0,0.5]\in [0,1]
}}
\end{center}
Hence, from Theorem~\ref{theorem:nearness}, we have
$\vec{V}f1\ \rnear\ \vec{V}f2$.
\end{proof}
\end{mdframed}

\vspace*{0.2cm}

\begin{remark}{\bf (Magiros Stable System Motions Condition).}\\
Let the extreme closeness stability condition Theorem~\ref{theorem:nearStableSystems} corresponds to a pair of vector fields \boxed{\vec{V}f1,\vec{V}f1: \vec{V}f1\ \rnear\ \vec{V}f1} derived from motion waveforms of a pair of physical systems.  In that case, the maximal \boxed{\lambda} different requirement would represent a pair of motion waveforms that are very stable.  That is, any small disturbance results in a small variation in the original waveform~\cite{Magiros1965}.  
\end{remark}

\vspace*{0.2cm}

\begin{example}{\bf (Vector Field Characteristics)}.\\
We have the followig characteristics for a vector field $(\vec{Vf},+)$ to work with, namely,
\begin{align*}
\vec{V}f &=\ \mbox{vector field in}\ \mathbb{C}.\\
Sg &= (\vec{Vf},+).\\
\varphi_1(Sg) &= 1\ \mbox{(tangent)} \vec{V}f.\\
              &\in (0,1)\ \mbox{(partly tangent)} \vec{V}f.\\
							&= 0\ \mbox{(non-tangent)}  \vec{V}f.\\
\varphi_2(Sg) &= 1\ \mbox{(normal)} \vec{V}f.\\
              &\in (0,1)\ \mbox{(partly normal)} \vec{V}.\\
							&= 0\ \mbox{(non-normal}\ \vec{V}f).\\
\varphi_3(Sg) &= \boxed{o(\vec{V}f)}\ \mbox{Vf size)}.\\
\varphi_4(Sg) &\left\{\vec{v}\subseteq \vec{V}f\right\}\in \left[0^o,\pm \frac{\pi}{2}\right]\ \mbox{(direction of $\vec{v}\in \vec{V}f$)}.\\
\varphi_5(Sg) &=
\abs{\abs{\varphi(\lambda_{\vec{V}f_1})}-\abs{\varphi(\lambda_{\vec{V}f_2})}}\in
[0,0.5] \Rightarrow \vec{V}f_1\ \rnear\ \vec{V}f_2.\\
\varphi_6(Sg) &= \mbox{(max)}\varphi(\lambda)\not{\in}\ \mbox{unit circle}\ \Rightarrow\mbox{unstable vector field}.\\
\varphi_7(Sg) &= \mbox{(max)}\varphi(\lambda)\in\ \mbox{unit circle}\ \Rightarrow\mbox{stable vector field}.\\
\Phi(Sg) &= \left\{\varphi_1(Sg),\varphi_2(Sg),
            \varphi_3(Sg),\varphi_4(Sg),\varphi_5(Sg),
						,\varphi_6(Sg),\varphi_7(Sg)\right\}.\mbox{\qquad \textcolor{blue}{$\blacksquare$}}  
\end{align*}
\end{example}

\vspace*{0.2cm}

\begin{figure}[!ht]
	\centering
\includegraphics[width=110mm]{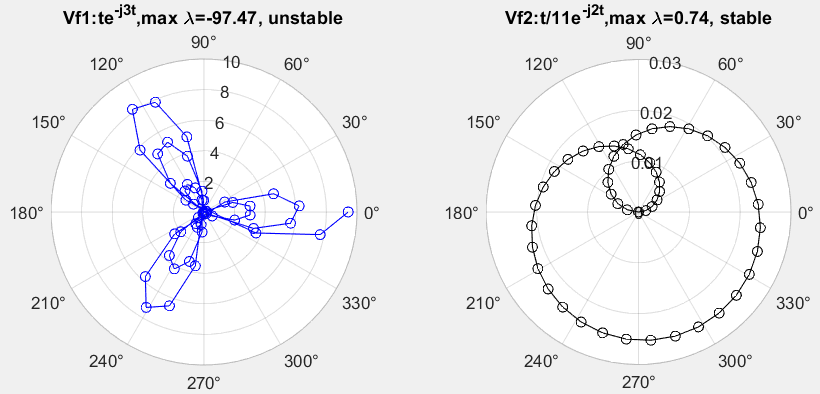}
	\caption{{\bf Case 1}: Characteristically non-near vector fields}
	\label{fig:Vf1Vf2}
\end{figure}

\begin{figure}[!ht]
	\centering
\includegraphics[width=110mm]{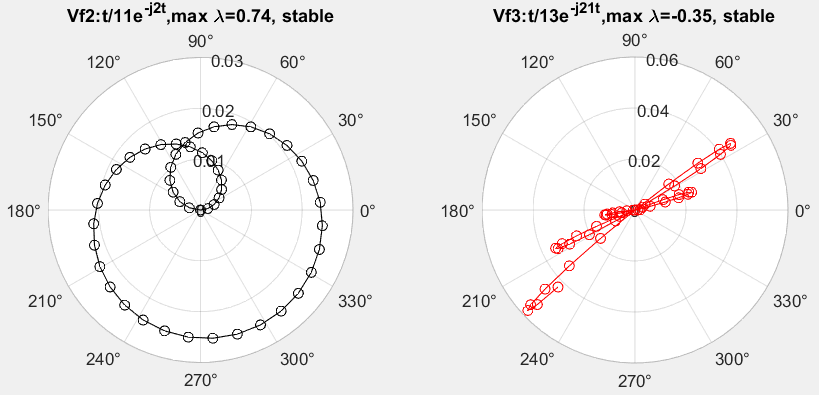}
	\caption{{\bf Case 2}: Characteristically near vector fields}
	\label{fig:Vf2Vf3}
\end{figure}

\begin{example}{(Characteristically Non-Near Vector Fields).}\\
In Fig.~\ref{fig:Vf1Vf2}, \boxed{\mbox{(not)}(Vf1\ \rnear\ Vf2)},
since 
\begin{align*}
\varphi_6(Sg_{Vf1}) & \mbox{(max)}\lambda=-97.47\ \Rightarrow\ \mbox{\bf unstable vector field}.\\
\varphi_6(Sg_{Vf2}) & \mbox{(max)}\lambda=0.74\ \Rightarrow\mbox{\bf stable vector field\qquad \textcolor{blue}{$\blacksquare$}}. 
\end{align*}
\end{example}  

\vspace*{0.2cm}

\begin{example}{(Characteristically Near Stable Vector Fields).}\\
In Fig.~\ref{fig:Vf2Vf3}, \boxed{Vf2\ \rnear\ Vf3},
since 
\begin{align*}
\varphi_5(Sg_{Vf2,Vf3}) & \abs{\abs{\varphi(\mbox{(max)}\lambda_{\vec{V}f_2}=0.74)}-\abs{\varphi(\lambda_{\vec{V}f_3}=-0.035)}}\in[0,0.5] \ \Rightarrow\mbox{\bf stable vector field}.\\ 
\varphi_6(Sg_{Vf2}) & \mbox{(max)}\lambda=0.74\ \Rightarrow\ \mbox{\bf stable vector field}.\\
\varphi_6(Sg_{Vf3}) & \mbox{(max)}\lambda=-0.35\ \Rightarrow\mbox{\bf stable vector field}.
\end{align*}
Also, observe that
\begin{align*}
\vec{v}_{left}\in \vec{V}f2 &= 0.02+0.1j\\
\vec{v}_{right}\in \vec{V}f3 &= 0.02+0.1j\\
\vec{v}_{left} &- \vec{v}_{right} = 0\\
               &\Rightarrow \vec{V}f2\ \rnear\ \vec{V}f3.\mbox{\qquad \textcolor{blue}{$\blacksquare$}}
\end{align*}

\end{example} 

\vspace*{0.2cm}

\begin{theorem}\label{theorem:charCloseAndProximal}{\bf (Characteristically Close Systems Are Proximally Close).}\\
Characteristic close systems are proximal.
\end{theorem}
\begin{proof}
This is an immediated consequence of the fundamental near systems Theorem~\ref{theorem:nearness}.
\end{proof} 

\begin{figure}[!ht]
	\centering
\includegraphics[width=120mm]{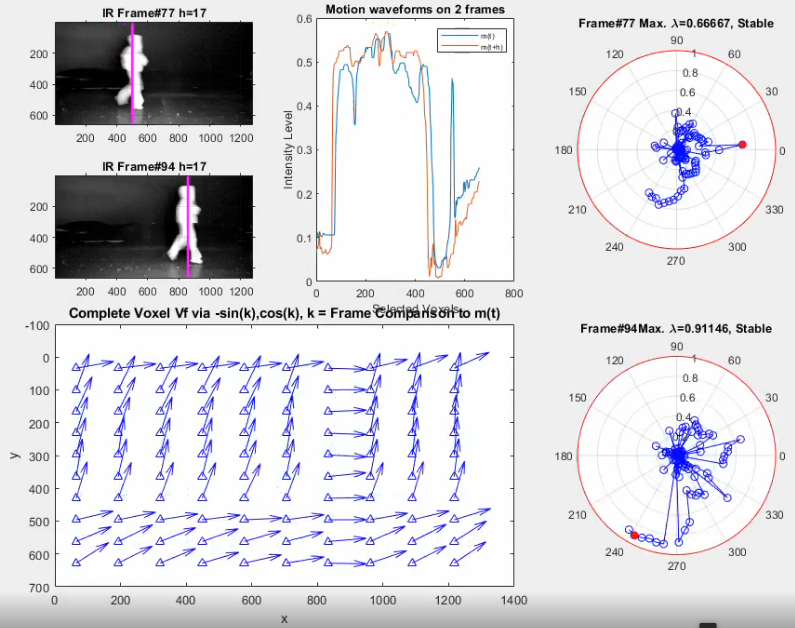}
	\caption{{\bf Case 1}: Characteristically near stable vector fields}
	\label{fig:nearStableFrames}
\end{figure}

\section{Application:characteristically near stable vector fields\\ on motion waveforms in infrared video frames}.
This section illustrates how to identify characteristically near motion waveforms in stable or unstable vector fields recorded in sequences of infrared video frames.  This application presents an advance over the method of evaluating motion waveforms in video frames that was introduced in~\cite{PetersLiyanage2024}.  In the following example, the vector fields emanate from seqiences of runner waveforms is recorded in frame sequences in infrared videos.  Be comparing the stability characteristics of the runner vector fields in pairs of video frames, we can then determine the overall stability of the runner. This approach carries over in assessing the characteristic closeness of the overall stability of the vector fields emanating from any vibrating system at different times. For simplicity, we consider only the maximum eigenvalues of the vector field in each videa frame.

\vspace*{0.2cm} 

\begin{example}{\bf (Case 1:  Pair of Characteristically Close Stable Vector Fields).}\\
In Fig.~\ref{fig:nearStableFrames}, contains a pair of characteristically near stable vector fields $\vec{V}f_{fr77}, \vec{V}f_{fr94}$ in frames 77 and 94.  Observe 
\begin{align*}
max\lambda_{fr77} & = 0.67,\\
max\lambda_{fr94} & = 0.91,\\
\abs{\abs{0.67}-\abs{-0.91}} & = 0.24 \in [0,0.5];\ \mbox{Hence, from characteristic}\ \varphi_5(Sg),\\
\vec{V}f_{fr77}\ & \rnear\ \vec{V}f_{fr94}.\ \mbox{\qquad \textcolor{blue}{$\blacksquare$}}
\end{align*}
\end{example}

\vspace*{0.2cm}

\begin{figure}[!ht]
	\centering
\includegraphics[width=120mm]{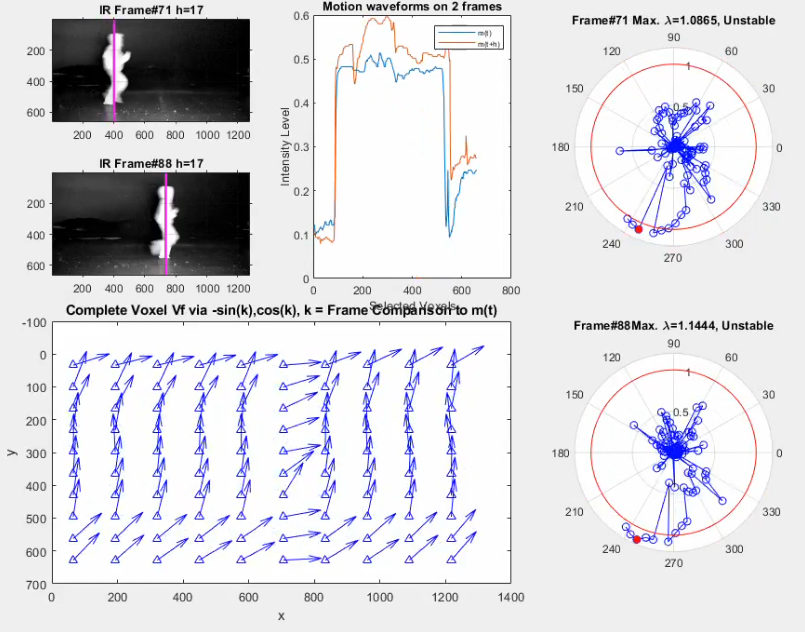}
	\caption{{\bf Case 2}:  Pair of Characteristically near unstable vector fields}
	\label{fig:nearNonStableFrames}
\end{figure}

\begin{example}{\bf (Case 2: Pair of Characteristically Close Unstable Vector Fields).}\\
In Fig.~\ref{fig:nearNonStableFrames}, contains a pair of unstable vector fields $\vec{V}f_{fr71}, \vec{V}f_{fr88}$ in frames 71 and 88.  Observe 
\begin{align*}
max\lambda_{fr71} & = 1.09,\\
max\lambda_{fr88} & = 1.44,\\
\abs{\abs{1.09}-\abs{1.44}} & = 0.35 \in [0,0.5];\ \mbox{Hence, from characteristic}\ \varphi_5(Sg),\\
\vec{V}f_{fr71}\ & \rnear\ \vec{V}f_{fr88}.\ \mbox{\qquad \textcolor{blue}{$\blacksquare$}}
\end{align*}
\end{example}

\vspace*{0.2cm}

\begin{figure}[!ht]
	\centering
\includegraphics[width=120mm]{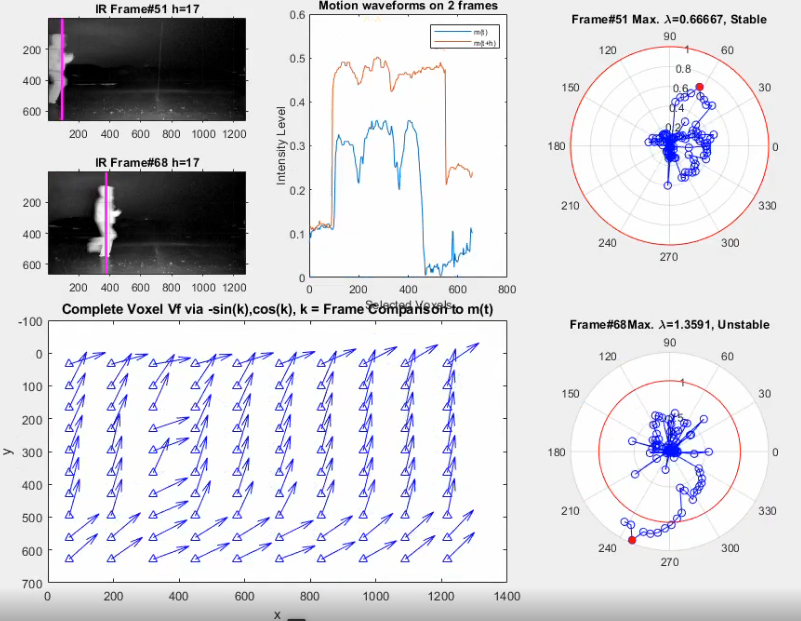}
	\caption{{\bf Case 3}: Pair of Characteristically near stable and unstable vector fields}
	\label{fig:nearStableUnstableFrames}
\end{figure}

\begin{example}\label{Case3:closeStableUnstableVfs}{\bf (Case 3: Characteristically Close Stable and Unstable Vector Fields).}\\
In Fig.~\ref{fig:nearStableUnstableFrames}, contains a stable vector field $\vec{V}f_{fr51}$ and unstable $\vec{V}f_{fr68}$ in frames 51 and 68.  Observe 
\begin{align*}
max\lambda_{fr51} & = 0.67,\\
max\lambda_{fr68} & = 1.36,\\
\abs{\abs{0.67}-\abs{1.36}} & = 0.69 \not{\in} [0,0.5];\ \mbox{Hence, from characteristic}\ \varphi_5(Sg),\\
\vec{V}f_{fr51}\ & \mbox{(not)}\rnear\ \vec{V}f_{fr68}.\ \mbox{\qquad \textcolor{blue}{$\blacksquare$}}
\end{align*}
\end{example}

\vspace*{0.2cm}

\begin{remark}{(Significance of Characteristically Non-Close Stable and Unstable Vector Fields in Case 3).}\\
Stable vector fields characteristically non-close to unstable vector fields are represented in Case 3 in Fig.~\ref{fig:nearStableUnstableFrames}. The vector fields in Example~\ref{Case3:closeStableUnstableVfs} have underlying systems that have the potential to be modulated to obtain a pair of charactistically close stable systems, since 
\begin{center}
\boxed{\boldsymbol{
\abs{\abs{0.67}-\abs{1.36}} = 0.69 \in [0,1]\ \mbox{(satisfies Theorem~\ref{theorem:nearness}).}
}}
\end{center}
That is, even though the vector field \boxed{\vec{V}f_{fr68}} is unstable in Case 3, it is characteristically close to the stable vector field \boxed{\vec{V}f_{fr51}} in Fig.~\ref{fig:nearStableUnstableFrames}. That characteristic closeness suggests the possibility of modulating the waveform slightly to change the vector field
\boxed{\vec{V}f_{fr68}} from unstable to unstable.

Unlike the temporal proximities of systems in the study in~\cite{HaiderPeters2021self-similar}, the characteristically close systems in Fig.~\ref{fig:nearStableUnstableFrames} are is different frames within the same video, but are separated by 10 frames and, hence, are not temporally close. The form of characteristic closeness introduced in this paper corroborates the results in~\cite{peters2025}.  Cases 1 and 2 illustrate the result in Theorem~\ref{theorem:charCloseAndProximal}, namely, characteristically close systems are proximal. \qquad \textcolor{blue}{$\blacksquare$}
\end{remark}

\section*{Acknowledgements}
We extend our thanks to Tane Vergili for sharing her profound insights concerning the proximity space theory in this paper.  In addition, we extend our thanks to Andrzej Skowron, Miros{\l}aw Pawlak, Divagar Vakeesan, William Hankley, Brent Clark and Sheela Ramanna for sharing their insights concerning time-constrained dynamical systems. In some ways, this paper is a partial answer to the question 'How [temporally] Near?' put forward in 2002~\cite{PawlakPeters2002}.

This research has been supported by the Natural Sciences \&
Engineering Research Council of Canada (NSERC) discovery grant 185986 
and Instituto Nazionale di Alta Matematica (INdAM) Francesco Severi, Gruppo Nazionale per le Strutture Algebriche, Geometriche e Loro Applicazioni grant 9 920160 000362, n.prot U 2016/000036 and Scientific and Technological Research Council of Turkey (T\"{U}B\.{I}TAK) Scientific Human
Resources Development (BIDEB) under grant no: 2221-1059B211301223.
\vspace{0.2cm}


\bibliographystyle{plain}

\end{document}